\documentclass{eptcs}
\providecommand{\event}{DSL 2011} 

\usepackage{amsmath,amssymb}
\usepackage{amsthm}
\usepackage{enumerate}
\usepackage{fancyvrb}
\usepackage{alltt,url}

\newlength{\dummylen}

\def\ind{\parindent}
\DefineVerbatimEnvironment%
  {program}{Verbatim}
  {xleftmargin=\ind,samepage=true}
\DefineVerbatimEnvironment%
  {programBR}{Verbatim}
  {xleftmargin=\ind,samepage=false}

\def\denseitems{
  \itemsep1pt plus1pt minus1pt
  \parsep0pt plus0pt
  \parskip0pt\topsep0pt} 

\def\OB#1{\ifmmode#1\else\mbox{$#1$}\fi}

\newcommand{\prog}[1]{\texttt{\mbox{#1}}}

\title{Adaptation-Based Programming in Haskell}
\author{Tim Bauer
\institute{Oregon State University\\ Corvallis, Oregon, USA}
\email{bauertim@eecs.oregonstate.edu}
\and
Martin Erwig
\institute{Oregon State University\\ Corvallis, Oregon, USA}
\email{erwig@eecs.oregonstate.edu}
\and
Alan Fern
\institute{Oregon State University\\ Corvallis, Oregon, USA}
\email{fern@eecs.oregonstate.edu}
\and
Jervis Pinto
\institute{Oregon State University\\ Corvallis, Oregon, USA}
\email{pinto@eecs.oregonstate.edu}
}
\def\titlerunning{Adaptation-Based Programming in Haskell}
\def\authorrunning{T. Bauer, M. Erwig, A. Fern, \& J. Pinto}
\begin{document}
\maketitle

\begin{abstract}
We present an embedded DSL to support adaptation-based programming (ABP)
in Haskell. ABP is an abstract model for defining adaptive values,
called \emph{adaptives}, which adapt in response to some associated feedback.
We show how our design choices in Haskell motivate higher-level combinators
and constructs and help us derive more complicated compositional adaptives.

We also show an important specialization of ABP is in support of
reinforcement learning constructs, which optimize adaptive values based 
on a programmer-specified objective function. This permits ABP users to easily
define adaptive values that express uncertainty anywhere in their programs.
Over repeated executions, these adaptive values adjust to more efficient
ones and enable the user's programs to self optimize.

The design of our DSL depends significantly on the use of type classes. 
We will illustrate, along with presenting our DSL, how the use of type 
classes can support the gradual evolution of DSLs.
\end{abstract}

\section{Introduction}
Programmers are often faced with the situation where it is not clear how to
best write a program that optimizes an objective of interest. For example,
consider designing an intelligent opponent for a real-time strategy game.
Computer-controlled opponents are typically quite weak and predictable
compared to an experienced human.  This is not too surprising since it is very
difficult for a programmer to anticipate all situations that will occur and to
specify the best course of action in each case.

As yet another example, consider trying to optimize the runtime of a
satisfiability solver or other type of constraint solver. There are many
decision points in such programs, and the best way to make the decisions, with
respect to runtime, depends very much on the distribution of inputs to the
program. Often this distribution is not known to the programmer and/or it may
change over the lifetime of the program. Even if the distribution were known,
the task of designing the best set of decision heuristics is quite daunting
and will often result in significant sub-optimality.  Unfortunately, standard
programming paradigms offer the programmer no choice but to completely specify
all such choices before program execution.

As another example, in the development of network control software it is
difficult to write complete programs that achieve close to optimal
performance. This is due to the dynamic, stochastic nature of networks leading
to uncertainty about the best values of parameters.

In this paper, we explore an embedded DSL to express \emph{adaptation-based
programming (ABP)}. In ABP, a programmer writes ``adaptive programs'' where
they are allowed
to explicitly specify their uncertainty by including ``adaptive values'' at the
program points where they do not know the best course of action.  In place of
specifying a concrete course of action, the programmer will be required to
specify an objective function that provides feedback about the quality of
program executions.  Given such an adaptive program, the adaptive values will
then be automatically adapted across program executions in an attempt to
optimize the specified objective. For example, the objective might be score in
a real-time strategy game, and the adaptive values might dictate which of some
number of strategies to employ in a specific game situation, or program
runtime might be the objective, and the adaptive values dictate choices among
different data structures and/or algorithmic choices. Provided that adaptive
programs can be optimized in an effective, automatic way, the ABP paradigm has
the potential to save significant development time and produce closer to
optimal program performance.

In the context of this paper, Haskell serves as an appropriate host language
for our embedded DSL. Haskell provides abstractions that facilitate the easy
experimentation with language ideas. Its type system forces us to be precise
in the description of language constructs while offering enough flexibility to
describe elements in their most general form. In particular, type classes
together with type functions \cite{SPC08} provide an elegant way of
formulating the notion of adaptive values.

Our DSL is defined around a type class and multiple functions that transform
and operate on instances of it. Programs of the DSL consist of instances of
this type class and allow the user to specify uncertainty.
We also provide template DSL programs for common patterns in the form of
generic instances such as adaptive pairs and functions as well as operations
supporting various patterns of evolution and adaptation.

As outlined in Section \ref{sec:relaw} there has been a small amount of prior
work in the Artificial Intelligence community on ABP under various names, most
notably partial programming \cite{Andre02}.  However, ABP has not yet been
studied as a general programming paradigm from a programming-language
perspective. It has been employed only by Artificial Intelligence experts on a
limited number of problems. 
This paper formalizes the ABP paradigm through an executable definition in
Haskell. This formalization is also likely to suggest unforeseen usage
patterns of ABP.
The main contributions of this paper include:

\begin{enumerate}[(1)]\denseitems
\item
Identification
of adaptive values as a foundation for adaptation-based programming and
their formalization through a corresponding Haskell type class.
\item
The definition of specific instances of adaptive
values, with intuitive interpretations, to be used as building blocks for
adaptive programs. In many cases these building blocks can draw on machine
learning theory to provide formal guarantees regarding their adaptation
behavior.
\item
Identification and definition of adaptable computation patterns
that are likely to arise in common practice.
\item
A formal convergence result for that provides a guarantee for the convergence
and optimality of a specific class of adaptive computations.
\item
A report on some practical experiments that illustrate the potential utility
of adaptive programming.
\end{enumerate}
The remainder of this paper is structured as follows.
In Section \ref{sec:aval} we introduce the notion of adaptive values and
define the interface to adaptive values through type classes.
The use of adaptive values to build adaptive computations is demonstrated in
Section \ref{sec:acomp}. We will identify adaptive computation patterns that
correspond to standard procedures in machine learning and those that are
likely to arise in some typical uses of ABP.
In Section \ref{sec:abehave} we present functions to monitor and control
adaptive computations.
In Section \ref{sec:convergence} we present a convergence result and discuss
the optimality of adaptive computations.
Section \ref{sec:case} provides some empirical results for the application of
ABP.
Related work is discussed in Section \ref{sec:relaw}, and
finally Section \ref{sec:concl} concludes and suggests future work.

\section{Adaptive Values}
\label{sec:aval}

The usual understanding of a value is that of a constant, unchanging object.
In contrast, an adaptive value can change over time. Changes to an adaptive
value are determined by feedback gathered from the context in which it is
used.

To facilitate a meaningful, controlled adaptation an adaptive value of type
\prog{v} needs to be represented, in general, by a somewhat ``richer'' type
\prog{a}, that is, type \prog{a} allows the extraction of values of type
\prog{v}, but also contains enough information to support interesting forms of
adaptation.

We call \prog{a} the \emph{representation type} and \prog{v} the \emph{value
type} of \prog{a}.
The adaptation is controlled by values of another type \prog{f}, called the
\emph{feedback type} of \prog{a}.
In the following we  call an adaptive value \emph{adaptive} for short to avoid
ambiguities between an adaptive value and the ``value of an adaptive value'',
that is, we simply say that \prog{x~::~a} is an adaptive and
\prog{value~x~::~v} is the value of (the adaptive) \prog{x} (\prog{value} will
be defined in Section \ref{sec:simple}).

In Section \ref{sec:simple} we describe the definition and examples of basic
adaptives, that is, adaptives defined directly on specific representation
types.
In Section \ref{sec:derived} we discuss obvious ways of obtaining compound
adaptives through derived instances for type constructors. A particularly
useful instance of this is the derived instance for function types that leads
to \emph{contextual adaptives} to be discussed in Section \ref{sec:context}.
In Section \ref{sec:nesting} we describe how to construct new adaptives
through nesting.

\subsection{Defining Adaptives}
\label{sec:simple}

The described concept of adaptives can be nicely captured by the
following Haskell type class.

\begin{program}
class Adaptive a where
  type Value a
  type Feedback a
  value :: a -> Value a
  adapt :: Feedback a -> a -> a
\end{program}
This class constitutes the core of our DSL: the operation \prog{value}
retrieves the current value from the representation, and the
function \prog{adapt} takes a feedback value and an adaptive and
produces a new adaptive. We represent points of uncertainty in our
program as instances of this class.

To define an adaptive representation type, a programmer has to provide an
instance definition for the class \prog{Adaptive}, which requires
%
%
\begin{itemize}\denseitems
\item implementations for the functions \prog{value} and \prog{adapt}, and
\item a definition of the corresponding value and feedback types
\end{itemize}
The value and feedback types are associated with the representation type
\prog{a} through the type functions \prog{Value} and \prog{Feedback},
which allows a large degree of flexibility in defining the adaptive
behavior \cite{SPC08}.

There are more things that we ultimately might want to store for adaptive
values for practical purposes (for example, statistics about usage, feedback,
and adaptation/adaptive behavior).
We will consider this aspect later in Section \ref{sec:abehave}.

As a simple example program we consider a form of incremental linear regression.
In particular, we want to learn the equation of a line $y = mx + b$
given a sequence of sample data points
$(x_1,y_1),(x_2,y_2),\ldots$.
The goal is to converge to an $m$ and $b$ that minimize the squared error
of predicting $y_i$ given $x_i$.

The adaptive for this example could be defined as follows.
First, we define the slope/intercept representation of lines.

\begin{program}
type Slope     = Double
type Intercept = Double

data Line  = L Slope Intercept
type Point = (Double,Double)
\end{program}
Based on this representation we can define the line adaptive as follows.

\begin{program}
instance Adaptive Line where
  type Value    Line = Line
  type Feedback Line = Point
  value = id
  adapt (x,y) (L m b) = L m' b'
        where m' = m + eta*x*(y - y0)
              b' = b + eta*(y - y0)
              y0 = m*x + b
              eta = 0.01
\end{program}
We can observe that the value of this particular adaptive
is just the same as the representation. 
The feedback is provided in the form of individual points,
each of which leads to an update of slope and intercept as defined by the
expressions for \prog{m'} and \prog{b'}. The value \prog{eta} represents the
learning rate, which is how much new inputs influence the adaptation.

As another example, consider the game of Rock-Paper-Scissors, in which two
players simultaneously choose one of three values \prog{Rock}, \prog{Paper}, or
\prog{Scissors}, trying to beat the opponent.

\begin{program}
data Move = Rock | Paper | Scissors
\end{program}
The winning move against each move is defined by the
following function \prog{win}.

\begin{program}
win :: Move -> Move
win Rock     = Paper
win Paper    = Scissors
win Scissors = Rock
\end{program}

It turns out that, given a fixed opponent, this game is a specific instance of
a so-called ``multi-armed bandit'' problem. This is a classic problem, first
described by Robbins \cite{Robbins52}, which captures the essential elements of
many experimental design problems, among others. The
problem can be viewed as modeling the process of playing a slot machine with
multiple arms, where each arm has an unknown distribution over random payoffs.
At each time step the player must select an arm to pull based on information
gathered from previous pulls, upon which a randomized return from the selected
arm is received. The goal is to develop an arm-pull strategy that maximizes
some measure of the expected payoff sequence over time, e.g.\ maximizing the
expected temporally-averaged payoff. In the case of Rock-Paper-Scissors with a
fixed opponent strategy, the arms correspond to the selection of either rock,
paper, or scissors, and the payoff reflects whether the selected move won or
lost against the selection of the opponent at that time step.

A good bandit strategy must balance the exploitation-exploration tradeoff,
which involves deciding whether to exploit the current knowledge and pull the
arm that currently looks best, or to explore other arms that have been tried
fewer times in the hope of discovering higher payoffs.

There are well known lower bounds on the performance of the best possible
strategy and bandit strategies that achieve those bounds asymptotically
\cite{Lai85}.  More recent work \cite{Auer02} has developed an upper
confidence bound (UCB) strategy, which was shown to achieve the lower bound
uniformly over all finite time periods.
Below, we describe a multi-armed bandit adaptive based on UCB.

In our representation of a multi-armed bandit we store a map that gives for
each arm how often it was pulled and the total rewards collected with it. The
representation is parameterized by the type used to represent the bandit's
arms.

\begin{program}
type Reward = Float
type Pulls  = Int
data Bandit a = Bandit (PlayMap a)
type PlayMap a = [(a,Pulls,Reward)]
\end{program}
The definition of the bandit adaptive has to return arm values (of type
\prog{a}) as values. The feedback is the arm that was pulled last together
with a reward that will be added to the total reward of that arm in the map.

We define the helper function \prog{updPM} to update the play map for a
given arm in some generic way.
\begin{program}
updPM :: Eq a => (ArmInfo a -> ArmInfo a) -> a -> PlayMap a -> PlayMap a
updPM _ _ []                 = []
updPM f x (a:as) | fst3 a==x = f a:as
                 | otherwise = a:updPM f x as

fst3 (x,_,_) = x
\end{program}

\noindent
With these definitions we can define a multi-armed bandit as an instance
of an adaptive.

\begin{program}
instance Eq a => Adaptive (Bandit a) where
  type Value    (Bandit a) = a
  type Feedback (Bandit a) = (a,Reward)
  adapt (a,r) (Bandit m) = Bandit (addReward r a m)
    where addReward :: Eq a => Reward -> a -> PlayMap a -> PlayMap a
          addReward x = updPM (\(a,p,r)->(a,p+1,r+x))
\end{program}
What remains to be defined is the \prog{value} method, for which we use the
UCB bandit algorithm. This approach first selects any arm
that has not been pulled before, which is achieved by the function
\prog{zeroPulls}, and otherwise selects the arm with the highest upper
confidence bound. This measure is defined for an arm $i$ that has been pulled
$n_i$ times and has a reward sum of $r_i$ as $r_i/n_i +\sqrt{\log n/n_i}$
where $n=\sum_{i}n_i$.

\begin{program}
  value (Bandit m) = a
      where ((a,_,_,):_) = zeroPulls ++ sortDesc ucb m
            zeroPulls    = filter ((==0) . pulls) m
            n            = fromIntegral (sum (map pulls m))
            ucb (_,p,r)  = r/ni + sqrt (log n/ni) where ni = fromIntegral p
            pulls (_,p,_) = p
\end{program}
The above function extracts arm \prog{a} by first choosing any arm that
has not been pulled (from \prog{zeroPulls}). If all arms have been
pulled, then it chooses the maximum value according to the UCB computation
given above. The function \prog{sortDesc} sorts a list in descending order of values as
obtained by the parameter function \prog{ucb}.

It is illustrative to note how the above UCB-based implementation of
\prog{value} manages the exploration-exploitation tradeoff. Assuming that all
arms have been pulled at least once, the decision is based on the upper
confidence bound, which is composed of two terms. The first term $r_i/n_i$ can
be viewed as encouraging exploitation since it will be larger for arms that
have been observed to be more profitable on average. Conversely, the second
term encourages
exploration since it grows with the total number of arm pulls, causing it to
overwhelm the first term if an arm has not been pulled very often. However,
the exploration term vanishes very quickly for an arm as its number of pulls
increases causing its evaluation to be based solely on its observed returns.
The result is that low-payoff arms tend to get fewer pulls than those with
higher payoffs over time, as desired.

The \prog{Bandit} instance is a generic operation in our DSL, it can be
utilized by many consumer programs. We illustrate one such use by coming
back to our Rock-Paper-Scissors example and instantiating the bandit as
an adaptive strategy for playing the game.
%

\begin{program}
type Strategy = Bandit Move
\end{program}
\begin{program}
initStrat :: Strategy
initStrat = Bandit [(m,0,0) | m <- [Rock, Paper, Scissors]]
\end{program}
We can use the following function \prog{score} to translate wins and
losses into numerical feedback.

\begin{program}
score :: Move -> Move -> Int
score m m' | win m  == m' = -1
           | win m' == m  =  1
           | otherwise    =  0
\end{program}
We can then pair \prog{initStrat} with other strategies and observe how it adapts
guided by the feedback values produced from \prog{score} applied to the moves
produced by \prog{value} and the opponent's move. We will do this in Section
\ref{sec:acomp} where we will identify and define adaptation computation
patterns that allow us to define applications (such as, line regression or
Rock-Paper-Scissors tournaments) that employ the defined adaptives.

One final note regarding the feedback employed for the multi-armed bandit: The
theoretical optimality result assumes the rewards are in the range $[0..1]$.
To adjust the \prog{Bandit} adaptive to the feedback produced by \prog{score}
we just needed to multiply the \prog{sqrt} term by 2.  However, in this example
the optimal behavior is not affected even if we don't scale the rewards since
all we are interested in is average reward.

\subsection{Derived Adaptives}
\label{sec:compound}
\label{sec:derived}
We define adaptation of generic structures in DSL by defining derived instances
of \prog{Adaptive}. This gives us instances of for adaptives for
many common patterns in adaptive programs.



As a first example, we define a derived instance of \prog{Adaptive} for pairs,
which realizes the parallel adaptation of two values in a synchronized
fashion.

\begin{program}
instance (Adaptive a,Adaptive b) => Adaptive (a,b) where
  type Value (a,b) = (Value a,Value b)
  type Feedback (a,b) = (Feedback a,Feedback b)
  value (x,y) = (value x,value y)
  adapt (u,v) (x,y) = (adapt u x,adapt v y)
\end{program}
One example use of this is the parallel adaptation of two
competing or even cooperating adaptive strategies in a game.
For instance, an AI or agent might have two goals that need to be
satisfied concomitantly. Then two \prog{Bandit}s, one adapting to
each goal automatically form a more complex agent that addresses both
with no additional programming.

Another example use of this particular construct will be given
in Section \ref{sec:acomp} where we can derive a co-evolution
computational pattern from a simple evolution pattern by using
this class instance definition.

We can also obtain an \prog{Adaptive} definition for lists. In this
definition, each adaptive's feedback value is used exclusively for that
adaptive.

\begin{program}
instance Adaptive a => Adaptive [a] where
  type Value [a] = [Value a]
  type Feedback [a] = [Feedback a]
  value = map value
  adapt = zipWith adapt
\end{program}
This definition can be generalized to any \prog{Functor} type constructor,
because we can easily define a corresponding \prog{fzipWith} function.


\subsection{Contextual Adaptives}
\label{sec:context}

A frequent scenario is to extend a given adaptive by context. For example, the
best arm to pull for a multi-armed bandit may depend on the time of day. Such
a context extension can be very conveniently achieved through the derived
\prog{Adaptive} instance for function types.
The idea is to turn an adaptive for some type \prog{a} into an adaptive for
functions from some context \prog{c} into \prog{a}. The value type of such an
adaptive function is a function from context into values of the original
adaptive \prog{a}, and feedback is given by feedback for \prog{a} enriched by
context information.
%
Contextual adaptive values are obtained in two steps. First, apply
the function to contextual information \prog{x}, and then extract
the value of that result. Adaptation based on a feedback \prog{(x,v)}
constructs an updated function that overrides input \prog{x} to map
to the adapted result of \prog{(f x)} with feedback \prog{v}. All other
inputs are delegated to the old function.
%
This definition illustrates that the functional adaptive essentially maintains
a number of separate copies of the original adaptive.

\begin{program}
instance (Eq c,Adaptive a) => Adaptive (c -> a) where
  type Value    (c -> a) = c -> Value a
  type Feedback (c -> a) = (c,Feedback a)
  value f = \x->value (f x)
  adapt (x,v) f = \y->if x==y then adapt v (f x) else f y
\end{program}
%
%
The definition for \prog{value} could be given more succinctly as
\prog{(value~.)}, but we think the above definition is easier to understand
and explains better what is going on.

This derived instance effective expands our DSL to support function types
transparently.

Note that this \prog{Adaptive} instance definition can be easily generalized
to a whole class of context type constructors, of which \prog{->} is one
example. A mapping type is another example, which might be preferable for
efficiency reasons.

%
%
As a concrete example we can add a player context to the multi-armed bandit
representing the Rock-Paper-Scissors player, which then allows the adaptive to
learn different strategies against different players.
\begin{program}
data Opponent = Jack | Jill deriving Eq

flexible :: Opponent -> Strategy
flexible = \_ -> initStrat
\end{program}
Note that this context-dependent strategy is obtained for free since it is
based on the automatically derived instance of \prog{Adaptive} for function
types.
For either player, the initial strategy is used, but as the function
receives feedbacks it will adapt more specialized strategies for each
player (input).

\subsection{Nested and Recursive Adaptives}
\label{sec:nesting}

Another way in which adaptives can be combined into more complex adaptives is
through nesting, that is, the value of one adaptive is another adaptive.  In
such a nested adaptive, value selection and adaptation happens on two levels.
While an ``ordinary'' adaptive represents an evolving decision, a nested
adaptive represents a sequence of such decisions.

To work effectively with nested adaptives it is not sufficient to simply place
one adaptive as a value into another one, because adaptation of the nested
adaptives would be impossible. The \prog{adapt} function for the outer
adaptive would simply adjust the selection of the nested adaptive. Although a
nested adaptive that is obtained by the \prog{value} function of the outer
adaptive can be adapted, there is no mechanism to put this changed adaptive
back into the outer one.

Therefore, we define a subclass of \prog{Adaptive}, called \prog{Dedaptive}, to
represent \emph{dependent adaptives}. These contain an extended value function
\prog{valueCtx}, which returns the value plus the context where it was found.
This context is a function that allows the value, or an adapted version of it
for that matter, to be put back into the containing adaptive.
The class also contains a function \prog{propagate} that allows the derivation
of feedback for the outer adaptive from feedback for the nested
one.
The additional first parameter of type \prog{a} serves two purposes: First, it
is needed to resolve the overloading of \prog{propagate}, and second it
provides a context of values to properly derive feedback, because in some
situations, the feedback type contains more than just an external value, but
also information related to the adaptive type.

\begin{program}
class (Adaptive a,Adaptive (Value a)) => Dedaptive a where
  valueCtx  :: a -> (Value a,Value a -> a)
  propagate :: a -> Feedback (Value a) -> Feedback a
\end{program}
Note that the dependency in nested adaptives goes both ways: The nested
adaptive depends as a value on the outer adaptive, while the outer adaptive's
adaptation is in part controlled, through \prog{propagate}, by the nested
adaptive.

As an example we can consider a nested multi-armed bandit. The nested bandit
could be a Rock-Paper-Scissors game or actually a gambling machine, while the
outer bandit could represent, for example, the decision at which time to play.

In the instance definition of \prog{Dedaptive}, the function \prog{valueCtx}
is based on the outer \prog{value} function to find the value. The context is
then simply obtained by isolating that value in a list and producing a
function that can insert an element in its place.
Since the feedback for a bandit of type \prog{a} is given by values of type
\prog{(a,Reward)}, we can produce feedback for the outer bandit simply by
pairing the reward provided for the nested one with the current value of the
outer one.

\begin{program}
instance (Eq a,Eq (Bandit a)) => Dedaptive (Bandit (Bandit a)) where
  valueCtx b@(Bandit m) = (a,\y->Bandit (xs++(y,p,r):ys))
           where a = value b
                 (xs,(_,p,r):ys) = break ((==a).fst3) m
                 fst3 (x,_,_) = x
  propagate b (_,r) = (value b,r)
\end{program}
%
We can now create a nested adaptive as follows.

\begin{program}
dependent :: Bandit Strategy
dependent = Bandit [(initStrat,0,0),(initStrat,0,0)]
\end{program}
It seems that \prog{dependent} is very similar to \prog{flexible}, and in
fact, we can simulate contextual adaptives by nested adaptives. However,
nested adaptives are more general since we can nest different adaptives (of
the same type) if we want, which is not possible for contextual adaptives.
This situation is reminiscent of the relationship between dependent and
independent products in type theory \cite{Tho91}.

Nested adaptives also raise the question of ``nested values'', that is, when we
want to get the value of a dedaptive, we in many cases do not want to have the
immediate value, which is itself an adaptive, but rather the ``ultimate''
value, that is, the value of the nested adaptive. This can be easily computed
by the function \prog{nestedValue}.

\begin{program}
nestedValue :: Dedaptive a => a -> Value (Value a)
nestedValue = value . value
\end{program}
%
%
%

\section{Programs for Adaptive Computation}
\label{sec:acomp}

The idea behind our adaptation DSL is the gradual evolution of values
to improve a programmatic solution to a problem.
This view requires that an adaptive computation, that is, a computation that
contains adaptive values, is performed repeatedly so that feedback, often
obtained from the results of the computation, is used to evolve the adaptives
employed in the computation.

Under this view, an adaptive computation has to contain (repeated) calls to
\prog{adapt} functions, and we can distinguish different adaptive computation
patterns based on the relationship of these adaptation steps with other
computations.

One of the most basic adaptation operations in our DSL is given by
the \prog{adapt} function itself, namely the one-step adaptation of an
adaptive. More complex patterns can be obtained by considering different
forms of repeated adaptation.

What is the result of an adaptive computation? Is it the final adaptive or the
trace of values that can be obtained from the list of all intermediate
adaptives, or both, or something else entirely?
For generality we define combinators for adaptive computation patterns to
return the list of all adaptives produced during the adaptation.
%
%
From this list we can easily obtain the final adaptive through the list
function \prog{last} or the trace of represented values through the function
\prog{valuesOf}, which is defined as follows.

\begin{program}
valuesOf :: Adaptive a => [a] -> [Value a]
valuesOf = map value
\end{program}
Other inspection and debugging functions for sampling or aggregating can be
added quite easily through ordinary list processing functions.

\subsection{Adaptation Combinators}
\label{sec:online}

One of the most basic adaptation patterns is to train an adaptive by a
list of training values analogous to supervised learning \cite{Bishop06}.
%
This is realized by the function \prog{trainBy} below.

\begin{program}
trainBy :: Adaptive a => a -> [Feedback a] -> [a]
trainBy = scanl adaptBy

adaptBy :: Adaptive a => a -> Feedback a -> a
adaptBy = flip adapt
\end{program}

\noindent
The {\tt scanl} function returns a list of all intermediate results
as a leftward fold is applied to a list. Here it will adapt an
initial adaptive in sequence and return the list (stream)
of all intermediate adaptives.

A more dynamic scenario is captured by the function \prog{evolve} that
uses its function parameter to compute feedback from the values of an
adaptive.

\begin{program}
evolve :: Adaptive a => (Value a -> Feedback a) -> a -> [a]
evolve f x = x:evolve f (x `adaptBy` (f (value x)))
\end{program}
The function \prog{evolve} represents a form of online learning
\cite{Bishop06} where the adaptive can be viewed as alternating between making
a decision (producing a value), getting feedback, and then adapting. The
bandit problem is a classic example of online learning, though there are many
other instances in the literature.
%


A generalization of \prog{evolve} is obtained by evolving two adaptives
in parallel where the values of both adaptives are the basis for
feedback to either one of the adaptives.
This definition makes use of the \prog{Adaptive} instance for pairs shown in
Section \ref{sec:derived}.
The function \prog{distr} makes the values of both adaptives available to
compute feedback.

\begin{program}
coevolve :: Adaptive (a,b) => (Value a -> Value b -> Feedback a,
                               Value a -> Value b -> Feedback b)
                               -> (a,b) -> [(a,b)]
coevolve = evolve . distr

distr :: (a -> b -> c,a -> b -> d) -> (a,b) -> (c,d)
distr (f,g) (x,y) = (f x y,g x y)
\end{program}
The adaptation pattern defined by \prog{coevolve} corresponds to the structure
of multi-agent reinforcement learning \cite{Littman94},
an area of reinforcement learning that studies situations where multiple
agents are learning simultaneously, possibly interacting with one another
either cooperatively or as adversaries.

As an example we consider the implementation of a Rock-Paper-Scissors
tournament. In addition to players, such as \prog{initStrat} described in Section
\ref{sec:simple}, we need functions to produce feedback values from the values
of two players. One such function is \prog{myScore}.

\begin{program}
myScore :: Move -> Move -> (Move,Reward)
myScore x y = (x,score x y)
\end{program}
Since different player adaptives might have other feedback types, we generally
need other functions as well. For example, a simple Rock-Paper-Scissors
strategy is to always play the move that wins against the last move of the
opponent.

\begin{program}
data BeatLast = BL Move

instance Adaptive BeatLast where
  type Value    BeatLast = Move
  type Feedback BeatLast = Move
  value (BeatLast m)   = m
  adapt m (BeatLast _) = BL (win m)
\end{program}
%
Recall \prog{coevolve} uses the value of both adaptives to produce the
corresponding feedback value for the adaptive. The function below
can be used to select the opponent's move from the previous round
and fits nicely with the above strategy.

\begin{program}
opponent'sMove :: Move -> Move -> Move
opponent'sMove _ y = y
\end{program}
Or consider a smarter strategy that plays the move that beats its
opponent's most frequently played move. This player maintains a
count that each move has been played.
\begin{program}
data Max = MP [(Move,Int)]
           deriving Show

instance Adaptive Max where
  type Value    Max = Move
  type Feedback Max = Move
  value (MP ms) = win (fst (maxWrt snd ms))
  adapt m (MP ms) = MP (updF m (+1) ms)
\end{program}
The function \prog{updF} updates a mapping in a list of pairs.
\begin{program}
updF :: Eq a => a -> (b -> b) -> [(a,b)] -> [(a,b)] 
updF x f []                     = []
updF x f ((y,w):as) | x==y      = (x,f w):as
                    | otherwise = (y,w):updF x f as
\end{program}
We can now define players as pairs of adaptive values plus their
corresponding feedback-producing functions.

\begin{program}
bandit = (initStrat, myScore)
beatLast = (BL Rock, opponent'sMove)
maxMv = (MP [(m,0) | m<-rps], opponent'sMove)
\end{program}
To be able to play strategies with their corresponding feedback function
against one another, we introduce the following tournament function.

\begin{program}
vs :: (Adaptive b, Adaptive a) =>
        (a, Value a -> Value b -> Feedback a)
     -> (b, Value b -> Value a -> Feedback b)
     -> [(a, b)]
(a,f) `vs` (b,g) = coevolve (f,flip g) (a,b)
\end{program}
%
%
Tournaments can then be played using \prog{vs} in the obvious way, for
example:

\begin{program}
beatLast `vs` maxMv
\end{program}
This example leads as expected to an overall victory for the \prog{maxMv} player.

\subsection{Recursive Adaptation}
\label{sec:recursive}

In Section \ref{sec:nesting} we have considered nested adaptives, in which
value selection and adaptation happens on two or more levels.  While an
``ordinary'' adaptive represents an evolving decision, a nested adaptive
represents a sequence of such decisions.

When the number of nesting levels is not fixed and not known in advance, it is
difficult to capture this computational pattern in a single combinator. In
that case, adaptation and value retrieval must be performed by individual
function calls that are integrated into the recursive structure of an adaptive
algorithm.

As an example we consider the problem of learning a combination of sorting
methods.
The idea is based on the observation that for a specific kind of lists, one
sorting method performs better than others.

To learn a combination of sorting algorithms we have to abstract some property
of lists and store costs or rewards for each sorting method under
consideration in a table indexed by that property. Since some sorting methods
are recursive, this will lead to a recursive adaptation process in which
potentially different sorting methods can be chosen based on the respective
properties of lists decomposed during the sorting recursion.

For simplicity we consider here the length of the list as a
property.\footnote{We actually use the square root of the list length to keep
the size of the table reasonable.}
We can build this adaptive table in two steps. First, we define an adaptive
for sorting methods, from which we can then create a table by adding the list
size as context, as demonstrated in Section \ref{sec:context}.

\begin{program}
data SortAlg = MSort | ISort
type Cost = Double
data Action = Action [(SortAlg,Int,Cost)]
\end{program}
The base adaptive for sorting algorithms has essentially the same structure as
a multi-armed bandit (see Section \ref{sec:simple}): It stores the number each
method was chosen together with the cost (representing running time).
Here we consider two methods, namely insertion sort and merge sort.

The \prog{Adaptive} instance definition for \prog{Action} is also very similar
to that of \prog{Bandit}. The only differences are that \prog{value} selects
the smallest entry (that is, the on average fastest sorting method) and that
\prog{adapt} updates a running average of costs via the \prog{updAvg} function.
We also choose any action not sufficiently explored attempted ($8$ is used
as cutoff to decide this).

\begin{program}
instance Adaptive Action where
 type Value Action = SortAlg
 type Feedback Action = (SortAlg,Cost)

 value (Action as)
    | null unexplored    = fst3 $ minWrt thd3 as
    | otherwise          = fst3 $ head unexplored
    where unexplored = filter (\a -> snd3 a < 8) as

 adapt (a,c) (Action as) = Action $ updF3 a
                             (\(a',f',c') -> (a', f' + 1, runAvg f' c' c)) as
\end{program}
The function \prog{runAvg} updates a running average, \prog{minWrt} selects
the minimum element with regard to some criteria in our case the average
time a sorting method takes, and \prog{updF3} remaps a specific triple
in a list.
\begin{program}
runAvg f c' c = (fd * c' + c) / (fd + 1)
  where fd = fromIntegral f

minWrt :: Ord b => (a -> b) -> [a] -> a
minWrt f = head . sortBy (\x y->compare (f x) (f y))

updF3 :: Eq a => a -> ((a,b,c) -> (a,b,c)) -> [(a,b,c)] -> [(a,b,c)]
updF3 x f []                      = []
updF3 x f (a:as) | x == fst3 a    = f a : as
                 | otherwise     = a : updF3 x f as
\end{program}

%
%
To support unlimited recursive adaptives, we use the adaptive as the state of
a state monad, which can then be used to thread adaptives through arbitrary
computations.
To facilitate the computation of actual timings for the given application, we
use a state monad transformer that encapsulates the \prog{IO} monad. The
following general definition of a Q-table \cite{RL-book} abstracts from the
concrete types for state/context (\prog{s}) and actions (\prog{a}).

\begin{program}
type QTable s a r = StateT (s -> a) IO r
\end{program}
Note that the state of the state transformer monad is a function
that represents a contextual adaptive.
For our example we have as an adaptive a function from list sizes to sorting
method adaptives.

\begin{program}
type Size = Int
type ASort r = QTable Size Action r

asort :: Size -> [Int] -> ASort [Int]
asort n xs =
  do let s = isqrt n
     q <- readTable
     let m = value q s
     t <- readTime
     ys <- case m of
              ISort -> isort n xs
              MSort -> msort n xs
     forceEval ys
     t' <- readTime
     modify (`adaptBy` (isqrt n,(m,t-t')))
     return ys
\end{program}
Adaptation sort takes as input a list \prog{xs} and its size \prog{n}, which
is used to select the best sorting method for the list.
First, the Q-table is
read from the state using the function \prog{readTable}, which is simply
another name for the state monad function \prog{get} that retrieves the state
of the monad. The value of the adaptive Q-table is the function that maps
sizes to sorting methods. Based on the selected sorting method \prog{m},
which is obtained by applying the function \prog{value~q} to the integer
square root of \prog{s}, we either sort using insertion sort or merge sort.
After forcing the evaluation of the result list \prog{ys}, we adapt the
Q-table using the monadic state updating function \prog{modify} before
returning the sorted list.

The recursively called sorting functions are also defined within the context
of the monadic adaptive \prog{ASort} since, at least \prog{msort}, has to
recursively sort sublists (of smaller size). That sorting task should be
performed using the currently best method for those lists, and it should also
adapt the information stored in the Q-table.

\begin{programBR}
isort :: Size -> [Int] -> ASort [Int]
isort _ xs = return (foldr insert [] xs)

msort :: Size -> [Int] -> ASort [Int]
msort n xs =
  if n<2 then
     return xs
  else
     do let k = n `div` 2
        let (us,vs) = splitAt k xs
        us' <- asort k us
        vs' <- asort (n-k) vs
        return (merge compare us' vs')
\end{programBR}
In Section \ref{sec:case} we report some concrete timing results for this
application, and we will present another application that is also based on
recursive adaptation.

\subsection{Transactional Adaptations}
\label{sec:transact}

The adaptive pattern operations considered so far all progressed in a very
fine-grained fashion, by tightly interwoven calls of \prog{value} and
\prog{adapt}. Although these patterns seem natural
there might be cases in which adaptation is less tightly controlled.
For instance it is often convenient for a multi-armed bandit may to have
several arm pulls per reward (\prog{adapt}) call.

To illustrate this consider the following alternative representation of our
multi-armed bandit, which stores in addition to the map the last pulled arm.

\begin{program}
type ArmInfo a = (a,Pulls,Reward)
type PlayMap a = [ArmInfo a]

data Bandit a = Bandit a (PlayMap a)
\end{program}
In order to maintain this representation we have to use a different feedback
type that distinguishes two kinds of feedback: either (a) an arm was pulled,
in which case the corresponding pull counter is increased and the arm is
remembered as the last one pulled, or (b) a reward for the last pulled arm is
delivered, which will be added to the total reward of that arm in the map.
These two different forms of feedback are captured in the following type.

\begin{program}
data Play a = Pull a | Reward Reward
\end{program}
This leads to a slightly different \prog{Adaptive} instance definition than
the one shown in Section \ref{sec:simple}.

\begin{program}
instance Eq a => Adaptive (Bandit a) where
  type Value    (Bandit a) = a
  type Feedback (Bandit a) = Play a
  adapt (Pull a)   (Bandit _ m) = Bandit a (incPulls a m)
  adapt (Reward r) (Bandit a m) = Bandit a (addReward r a m)
\end{program}
The function \prog{incPulls} increments the number of pulls of the given
arm in the map, \prog{addReward} adds reward for a given arm.
The definition of \prog{value} remains unchanged and still uses
the UCB algorithm previously described.

Now consider what happens if we want to implement a Rock-Paper-Scissors
strategy on the basis of this representation and play it against some other
strategy.
The problem is that it now takes \emph{two} adaptation steps, a \prog{Pull} of
an arm and a \prog{Reward} for it, to make a meaningful adaptation transition
in the sense of machine learning.
Therefore, we need some form of ``big-step'' adaptation that can for this
example be derived from the adaptive's feedback as follows.

\begin{program}
bigStep :: Eq a => (a,Reward) -> Bandit a -> Bandit a
bigStep (x,r) b = b `transBy` [Pull x,Reward r]

transBy :: Adaptive a => a -> [Feedback a] -> a
transBy = foldl adaptBy
\end{program}
The point to observe is that we have converted a value of type
\prog{Feedback~a} into a function of type \prog{a~->~a}, which means that the
big-step adaptation pattern that  corresponds to \prog{trainBy} takes a list
of such functions instead of feedback values.

\begin{program}
transformBy :: a -> [a -> a] -> [a]
transformBy = scanl (flip ($))
\end{program}
%
Consider, for example, an adaptation of the following form.

\begin{program}
initStrat `trainBy` xs
\end{program}
The corresponding adaptation for the changed adaptive could be implemented
using \prog{transformBy} in the following way.
Here \prog{stratB} is the initial bandit value, defined in the same way as
\prog{initStrat} for the new \prog{Bandit} type.

\begin{program}
stratB `transformBy` map bigStep xs
\end{program}

As for \prog{trainBy} we can also produce a big-step version of
\prog{coevolve} by generalizing the type of the argument functions. The result
is a function that adapts two adaptives based on big-step adaptation
parameter functions that have access to both current adaptives.

\begin{program}
cotransform :: Adaptive (a,b) => 
               (a -> b -> a,b -> a -> b) -> (a,b) -> [(a,b)]
cotransform (f,g) (x,y) = (x,y):cotransform (f,g) (f x y,g y x)
\end{program}
An example would be the definition of a Rock-Paper-Scissors tournament for
adaptives as defined at the beginning of this section.

% ...
%
%
%
%
% ...

\section{Monitoring Adaptation Behavior}
\label{sec:abehave}
The lifetime of adaptive programs can often be split into two major phases:
(i) a \emph{learning} or \emph{adaptation phase} in which adaptives adapt
(significantly) and (ii) a \emph{stable phase} in which no or only minor
adaptations occur.
%
%
It might be desirable, for example if we are training an adaptive with
predefined feedback, to be able to detect this transition.

To determine whether an adaptive program is stable requires to monitor the
adaptives. To this end, we define a type \prog{Monitor} and a
corresponding function \prog{monitor} to produce observations about the
adaptation behavior.

\begin{program}
type Monitor a b = [a] -> b

monitor :: Adaptive a => Monitor a b -> [a] -> [b]
monitor m = map m . inits
\end{program}
The function \prog{inits} produces the list of all prefixes of a given list.


Here is an example monitor that ensures that a particular property holds for
the values of the $k$ last adaptives produced in an adaptation.

\begin{program}
ensureLast :: Adaptive a => Int -> ([Value a] -> Bool) -> Monitor a Bool
ensureLast n p xs = length xs >= n &&
                    p . map value . take n . reverse $ xs
\end{program}
%
%
A very simple example property to monitor is whether all the values in a list
are the same.

\begin{program}
allEq :: Eq a => [a] -> Bool
allEq []     = True
allEq (x:xs) = all (==x) xs
\end{program}
This property can be used to define a simple convergence criterion as follows.

\begin{program}
convergence :: (Adaptive a,Eq (Value a)) => Monitor a Bool
convergence = last 3 allEq
\end{program}
Using monitors we can define adaptation combinators that are controlled by the
monitors.

\begin{program}
until :: Adaptive a => [a] -> Monitor a Bool -> [a]
until xs = shiftMonitor ([],xs)

shiftMonitor :: ([a],[a]) -> Monitor a Bool -> [a]
shiftMonitor (xs,[]) m = if m xs then xs else []
shiftMonitor (xs,y:ys) m | m xs      = xs
                         | otherwise = shiftMonitor (xs++[y],ys) m
\end{program}
With \prog{until} we can now define self-controlling adaptations that adapt
until a certain criterion, such as \prog{convergence}, is met.

As a concrete example, consider again the linear regression scenario. We can
adapt a line \prog{l} using a list of points \prog{ps} until the last two
lines in the approximation sequence are close enough together, that is, their
difference in slope and intercept is smaller than a specific threshold.

\begin{program}
(l `trainBy` ps) `until` ensureLast 2 areClose

areClose :: [Line] -> Bool
areClose [L m b,L n c] = max (abs (m-n)) (abs (b-c)) <= 0.001
\end{program}

\section{Convergence and Optimality}
\label{sec:convergence}

One of the primary motivations for the ABP framework is to allow for
programs to automatically optimize their performance relative to
programmer-specified objectives. Thus, it is important to understand conditions
under which an adaptive program might converge to an optimal or approximately
optimal solution. Convergence of an adaptive program depends on statistical
properties of the adaptives and program inputs, as well as the structure of
the program. In general, understanding convergence issues is quite complex,
and we leave the general problem as future work. Instead below we take an
initial step in this direction for a particular type of adaptive used in a
restricted, but powerful, class of adaptive programs which we will call
\emph{single adaptive recursive functions (SARFs)}.

The definition of the SARF class of functions is inspired by the structure of
the adaptive sorting example. Specifically, SARFs are recursive functions
that possibly  call other functions, with the following three restrictions: 

\begin{enumerate}\denseitems
\item
There is a single adaptive in the entire program.
\item
The value of the adaptive
is used only once in the main function and used nowhere else in the program.
\item
For any instance of the adaptive and any function input, the function will
terminate in a finite amount of time (i.e. no infinite recursion).
\item
The
feedback is a numeric cost that is a function of the computation that took
place during the function call.
\end{enumerate}
Note that the adaptive sorting function is a
SARF, where the feedback corresponds to the time required for the function to
complete execution.

In order to study the convergence of SARF programs, we must first formalize
the notion of optimality.  For this purpose, we define optimality with respect
to an unknown but fixed probability distribution $D$ over possible inputs to
the SARF. For simplicity, we will also assume that there exists a finite upper
bound such that the probability of inputs larger than the bound is zero. For
example, in the adaptive sorting example, $D$ might be a distribution over
random lists up to some maximum size. Given a SARF $P$ and a distribution over
inputs $D$, we define $C(P,D)$ to be the expected cost of executing $P$ on
inputs drawn from $D$, where cost is as defined in $P$. We are interested in
adaptation processes such that $P$ will eventually achieve the optimal cost
with high probability after some number of adaptations. In particular, given
an initial $P$, we consider applying $P$ to a sequence of inputs drawn from
$D$, each time allowing it to adapt, ideally resulting in a version of $P$
that achieves the optimal expected cost.

Naturally, convergence depends on the choice of adaptive in a SARF. One option
would be to use a contextual bandit adaptive. It turns out that analyzing the
convergence of the resulting $P$ is quite complex due to the fact that the
quality of the decisions at higher levels of the recursion depend on decisions
at lower levels of the recursion, which would always be adapting in the case
of contextual bandits. We conjecture that convergence can be guaranteed for
the contextual bandit case, however, we leave it as future work. Here
we define a restricted class of adaptives, called \emph{principled adaptives},
that allows for an easier convergence proof.

Intuitively this adaptive will attempt to ``learn'' the quality of the actions
in a context from the bottom up with respect to the depth of the recursion.
Roughly speaking, the principled adaptive can be viewed as first learning the
quality of the actions for contexts corresponding to the recursion base cases.
Next, fixing those contexts to the best decision, learning proceeds to
contexts that are one level removed from a base case. Here the quality of each
action is judged conditioned on the fact that the base case decisions are
fixed and ideally optimal. Once these action qualities are learned well enough
they are in turn fixed and learning proceeds one level higher. Note that under
this strategy action qualities for a context are only learned, or updated,
when lower level decisions are fixed, rather than when the lower level is also
adapting.

More formally, the principled adaptive is similar to the adaptive sorting
adaptive in that it is based on a Q-table.  The key difference is the way that
it computes values and does the adaptation. The principled adaptive is
parameterized by an integer $t$, which we will call the learning threshold.
Our convergence results will specify sufficient values of this parameter.
First we introduce some terminology. We say that a context-action pair is
\emph{stable} if it has been updated at least $t$ times. We say that a context
is stable if all of its actions are stable. Intuitively, we will think of
stable context as one where we are quite confident that we know the correct
action to select. Given these definitions we can now specify the value and
adaptation function of the principled adaptive.

The value function returns the action that minimizes the Q-table (best action)
if the context is stable, and otherwise selects the first unstable action in
the context. According to this definition, a greedy action is only returned in
a particular context, if all of the actions in that context have been updated
at least $t$ times. In this sense, the value function is aggressive about
exploring all of the actions equally before settling on one of them. Now
suppose that our SARF has computed a value for the current input, resulting in
a cost that is used as feedback to the \prog{adapt} function of the principled
adaptive.  There are two cases that are handled. In the first case, if any of
the recursive calls involved a context that was not stable or the current
context-action pair has already been updated $t$ times, then no update is
performed.  Otherwise, if all recursive calls involved stable contexts then
the Q-table is updated based on the feedback for the appropriate context and
action that was selected. This involves updating the average cost observed for
the context-action pair.

Given a SARF with a principled adaptive we must make two assumptions to
guarantee convergence. First, we must assume that there is an ordering of
the contexts of the adaptive that strictly descends with the level of
recursion. That is, given a recursive call in context $c$, all lower-level
calls must correspond to contexts that are ordered lower than $c$. In this
case, we say that the adaptive has the \emph{descending context property}.
This property holds in the adaptive sorting example where list length is the
context and each recursive call decreases the list length. It also holds for
our budget optimization problem in Section \ref{sec:case}. The second
assumption, is the \emph{call-invariant cost assumption}. This assumption
means that for any context $c$, if all contexts ordered below $c$ are assigned
fixed actions, then the distribution over costs observed when taking each
action in $c$ is independent of the decisions in the higher-level contexts. In
other words, the recursive path taken to a particular context does not
influence the costs of the choices at that context.

We can now present the convergence result. First note that if there are $N$
contexts and $A$ actions, then the maximum number of Q-table updates that the
principled adaptive will perform is $t\cdot N\cdot A$. What our convergence
result states is that with high probability on inputs where no adaptation
occurs the optimal decisions will be made. In the following $\epsilon$ denotes
the minimum difference across all contexts between the expected cost of an
optimal action and the expected cost of the second best action with respect to
inputs drawn from $D$, and $\delta$ is the (user-selected) failure probability
of not being optimal.

\newtheorem{theorem}{Theorem}
\begin{theorem}
Let $P$ be a SARF with principled adaptive that has descending contexts and
call-invariant costs with respect to a fixed input distribution $D$. If the
learning threshold $t>4\epsilon^{-2} \ln\frac{NA}{\delta}$, then with
probability at least $1-\delta$, $P$ will behave optimally on all but $t\cdot
N \cdot A$ inputs drawn from $D$.
\end{theorem}

\begin{proof}
(Sketch) First note that the descending contexts assumption can
be used to show that whenever an input results in a computation that goes
through a non-stable context that some context-action pair in the Q-table will
be updated. This means that in all other cases the input will only go through
stable contexts and thus all actions selected for those inputs are ones that
are best as judged by the Q-table. Below we argue that with probability at
least $1-\delta$ that all such actions will correspond to optimal actions,
which will compute the proof.

Let $Q^*(s,a)$ be the optimal expected cost of action $a$ in context $s$ and
let $Q(s,a)$ be the cost of the adaptive for $s$ and $a$. Recall that
whenever the value of $Q(s,a)$ is updated that the decisions in all
lower-level contexts have been fixed. Consider the case where all of the lower
level contexts are fixed at optimal decisions. Under this assumption and the
call-invariance assumption we can use a Chernoff bound\footnote{%
Given a real-valued random variable $X$ bounded in absolute value by
$X_{\mathrm{max}}$ and an average $\hat{X}$ of $w$ independently drawn samples
of $X$, the additive Chernoff bound states
that with probability at least $1-\delta$, 
$|E[X]-\hat{X}| \leq X_{\mathrm{max}}\sqrt{-\ln\delta/w}$.}
%
to get that after $t$
updates of $Q(s,a)$, with probability at least $1-\delta'$,
%
%
$|Q^*(s,a)-Q(s,a)|\leq \sqrt{-\ln\delta'/t}$.
By setting $\delta'=\frac{\delta}{CA}$ we guarantee that the bound will hold
over all $s$ and $a$ with probability $1-\delta$. Using this value and the
bound on $w$ from the proposition we get that $|Q^*(s,a)-{Q}(s,a)|<
\frac{\epsilon}{2}$ with probability at least $1-\delta$, which by the
definition of $\epsilon$ implies that no sub-optimal action will be ranked
higher than an optimal action.

In the above we assumed that whenever $Q(s,a)$ is updated all lower-level
context made optimal decisions. Using a proof by induction on the ordering of
the contexts it is easy to use the above argument to prove that in each
context the optimal action will look best with probability at least
$1-\delta$, which completes the proof.
\end{proof}

\section{Empirical Results}
\label{sec:case}
Here we present empirical results for the application of ABP to two well
known problems, RL has been previously applied to: Sorting
\cite{Lagoudakis00} and budgeted optimization \cite{Ruv08}. Our framework is
able to naturally capture both problems, allowing for most of the details of
the adaptation process to be hidden from the programmer.

\textbf{Adaptive Sorting.} Prior work \cite{Lagoudakis00} on adaptive sorting used
RL to learn to choose between quicksort and insertion sort at each recursion
point based on the length of the list. The learned program showed small gains
in average runtime over pure quicksort and insertion sort.
We implemented adaptive sort using the structure shown in Section
\ref{sec:recursive} to learn a mixed strategy of insertion sort and
merge sort.\footnote{We used a tree-based map as a contextual adaptive instead
of functions for performance reasons.} 
We trained the algorithm on lists of integers of lengths up to $10000$.
The learned policy found a cutoff of just above $300$: For lists smaller than
that, insertion sort was faster, whereas for lists longer than the cutoff,
merge sort was faster.
Next we tested
our learned algorithm policies of just merge sort (no cutoff) and
merge sort with cutoffs off $10$ and $1000$.
The learned algorithm was considerably faster than just mergesort
with the other cutoffs we tested. For lists of size $10000$, we see
a speedup of between $1.6$ and $2.6$. Against merge sort with no
cutoff, the learned algorithm is $20$ times faster.

%
An important observation was that the cutoff learned only applies in the
environment it was learned.  That is, when we were learning the cutoff we were
accessing the system timer and modifying our adaptives as we sorted lists.
This overhead is necessarily included in the time we record to sort a sublist
(in \prog{asort}).
But if we sort in an environment without this overhead, the learned
cutoff does not apply, and a different one is optimal. In fact, tests
showed a very low cutoff (perhaps none) was fastest if there is
no overhead.

Whenever using time as a cost or reward, one must consider the fact that the
timing observations influence the results.
Although our adaptive framework is fairly
fast and efficient, the action being timed (sorting in our case) must be
significant compared to this overhead. In this sorting domain, the time to
sort a list was only significant for larger lists.

\textbf{Adaptive Budgeted Optimization.} We consider budgeted optimization where
we are given a function $f:\mathbb{R}^n \rightarrow \mathbb{R}^m$ and wish to
find the value of $x$ that minimizes the ``squared loss'' function $
L(x)=|f(x)|^2$. Furthermore, we are given a budget $B$ on the maximum
number of times that we are allowed to evaluate $f$ during the optimization
process. This situation of budgeted (or time-constrained) optimization occurs
mostly due to real-time performance requirements (for example in computer
vision and control problems).

We consider applying the standard Levenberg-Marquardt (LM) algorithm \cite{Lev44}
to this problem. LM is an iterative optimization algorithm that
starts at a random location $x_0$ and on each iteration evaluates the function
at the current $x_i$ and computes a new $x_{i+1}$. LM uses a mixture of
gradient descent and Gauss-Newton optimization to compute $x_{i+1}$. The
details of this computation are not particularly important other than the fact
that a key component of the algorithm is that each iteration must decide how
to best blend gradient descent and Gauss-Newton, which is done by specifying a
blending parameter $\lambda$. Marquardt \cite{MAR63} proposed a simple way to
modify $\lambda$ by increasing $\lambda$ by a particular factor (putting more
weight on gradient descent) when the previous iteration increased the loss,
and decreasing $\lambda$ otherwise (giving more weight to Gauss-Newton). This
$\lambda$ control works well and can be found in most implementations.

In \cite{Ruv08}, the authors apply reinforcement learning (RL) methods to
learn a controller for $\lambda$ and show that it is possible to obtain a
small improvement with respect to reduction in loss compared to the standard
$\lambda$ control. We applied our framework to this problem using recursion to
implement the iteration, resulting in an adaptive that has seven actions:
three actions either increase, decrease, or do not
change $\lambda$ while keeping the value of $x$ produced by the previous
iteration; three actions that are similar but discard the new value of $x$;
and an action that resets the value of $x$ to the best one seen so far. The
adaptive's context is a triple $(b,h_1,h_2)$ where $b$ is the remaining budget
and $h_1$, $h_2$ are indicators that encode whether the loss improved on the
previous step and two steps back, respectively.

We tested the adaptation behavior of our program on three classic benchmark
problems \cite{UCTP}: (1) Rosenbrock, (2) Helical Valley, and (3) Brown \&
Dennis function, using a budget of 5 function evaluations. We adapted the
function for $3 \cdot 10^6$ for different starting points, each one
applying the adaptive procedure to
one of the functions drawn at random starting at a random initial $x$ value
from $[-10,10]^n$ (where the dimension $n$ is 2 in the case of Rosenbrock, 3
for Helical Valley, and 4 for Brown \& Dennis).  After training we evaluated
the averaged scaled reduction in loss (ASRL) of the resulting procedure over
$10^5$ initial $x$ values for each function, where ASRL is simply the average
across all runs of the reduction in loss divided by the loss of the initial
value of $x$.
%
%
Our results indicate a reduction in ASRL over the standard LM for two
functions: Brown \& Denis by $0.01$ and Rosenbrock by $0.004$. For Helical Valley,
the ASRL increases by approximately $0.013$ over the standard LM. The modest
improvements are a due to the fact that standard LM is close to optimal for
these functions given our budget, leaving little room for improvement.
However, the fact that the adaptive program was able to learn to match the
standard LM performance is a great success; it demonstrates the effectiveness of
the adaptation in this particular example and indicates that the ABP idea works
well in practice.
 

\section{Related Work}
\label{sec:relaw}

In \cite{Bauer11} we present a slightly different view of adaptive programming.
There we viewed ABP in the context of a popular object-oriented language
in a much more focused and limited form. For instance, the feedback
type is fixed to be a numeric reward rather than an arbitrary type.
The goal of that work was to support non-expert programmers and shield them
from some of the complexities inherent in any adaptive system.
Conversely, this work's goal is to understand how adaptive values interact with
each other and form adaptive programs in general.

Acar's work on Self-Adjusting Computation \cite{Acar} presents a different
view of adaptive programming where the goal is to produce programs that
adjust automatically in response to any external change to their state.
The aim of this work is more in support of dynamic (online) algorithms
and incremental data structures instead of the feedback-driven program
optimization we present.

The ABP paradigm is inspired by recent work under
the name partial programming in the field of reinforcement learning (RL). RL
\cite{RL-book} is a subfield of artificial intelligence that studies
algorithms for learning to control a system by interacting with the system and
observing positive and negative feedback. RL is intended for situations where
it is difficult to write a program that implements a
high-quality controller, but where it is relatively easy to specify a feedback
signal that indicates how well a controller is performing. Thus,
pure RL can be viewed as an extreme form of ABP where the non-adaptive part of
the program is trivial, requiring the RL mechanisms to solve the full problem
from scratch. As such, successful applications of RL typically require
significant expertise and experience. It is somewhat of an art to formulate a
complex problem at the appropriate abstraction level so that RL will be
successful.

The inherent complexity of pure RL led researchers to develop different
mechanism for humans to provide natural forms of ``advice'' to RL systems, e.g.
in the form of a set of rules that specify hints about good behavior in
various situations \cite{Maclin05}, or example demonstrations of good behavior
by a domain expert \cite{Abbeel04}. However, these forms of advice still
require an RL expert who is very familiar with the underlying algorithms for
their successful application. In addition, the expressiveness of the types of
advice that can be provided are quite limited, particularly in comparison with
programming languages.

The desire to increase the expressiveness of advice provided to RL systems has
resulted in research on hierarchical reinforcement
learning \cite{Dietterich98}. Here a human specifies
behavioral constraints on the desired controller, or program, to be learned
in the form of sub-task, or sub-procedure, hierarchies. The
hierarchies specify potential ways that the high-level problem can be solved
by solving some number of sub-problems, and in turn how those sub-problems can
potentially be broken down and so on.  Not all of the possibilities specified by the hierarchies
will be successful or optimal, but the space of possible controllers can be
dramatically smaller than the original unconstrained problem. Given these
constraints, RL algorithms are often able to solve substantially more complex
problems.

Provided with enough constraints the hierarchies described above can be viewed
as defining programs. This idea was made explicit under the name partial
programming, where a simple language based on hierarchical state machines was
developed to provide guidance to an RL agent \cite{Andre01}. This language was
soon replaced by the development of ALISP \cite{Andre02}, which was a direct
integration of RL with LISP. The key programming construct that ALISP adds to
LISP is the choice point, which is qualitatively similar to an adaptive value
in our framework. The primary focus of work on ALISP has been to develop
adaptation rules for choice points and to understand the conditions under which
learning would be optimal in the limit of infinite runs of the program in an
environment.
%
%

Genetic Programming (GP) is a biologically-inspired approach for optimizing
programs based on a type of randomized search. Thus, like RL applied to ABP,
GP aims to optimize some objective over program runs. However, unlike RL,
GP does not typically exploit the sequential nature of program executions
during the optimization process. Rather, GP is a more generic black box
optimization approach, which typically ignores all aspects of the program
execution, except for the final returned objective value. In this sense,
RL is arguably a more appropriate formalism for ABP since it is specifically
designed for sequential decision making problems. 

A more recent proposal for an adaptive programming language
is A$^2$BL \cite{Simpkins08}, which integrates RL with the agent behavior
language (ABL). The proposal for A$^2$BL can be viewed as an instance of ABP
for a language that is specialized to behavioral-based programming of software
agents. Few details concerning a concrete syntax, implementation, and learning
rules are currently available.

Our work is also inspired by prior work on partial programming. To date, work on
ABP or partial programming has been largely orthogonal to the main
contributions of this paper. Most importantly, the existing work has not
resulted in a well-founded notion of ABP from a programming language
perspective, which has left many open issues regarding the pragmatics and
properties of adaptive programs. Our work is the first to formalize ABP in a
declarative language and to define primitive ABP elements, their combinations,
and programming patterns.

\section{Conclusions and Future Work}
\label{sec:concl}
We presented a generic embedded DSL in Haskell for describing adaptation-based
computations, which is based on the concept of adaptive values.
We demonstrated how standard machine learning scenarios and more general
adaptive programs can be captured via simple computational patterns.
Initial experiments demonstrated the potential of ABP. 
The main goal was to understand what constructs a DSL for adaptive
programming should support and what programming patterns we can identify in
adaptive programs.
In future work we will investigate more formal properties of ABP.
In particular, we want to identify laws for optimizing adaptives with
regard to convergence rate. Furthermore, we intend to extend the language
to patterns found in larger adaptive programs with the aim of solving
harder problems.

The implementation described in this work is available at \cite{haskell-abp}
to the curious reader.

\section*{Acknowledgments}
This work is supported by the National Science Foundation under the grant
CCF-0820286 ``Adaptation-Based Programming''.

\bibliographystyle{eptcs}

\bibliography{abp}

\begin{thebibliography}{10}
\providecommand{\bibitemdeclare}[2]{}
\providecommand{\urlprefix}{Available at }
\providecommand{\url}[1]{\texttt{#1}}
\providecommand{\href}[2]{\texttt{#2}}
\providecommand{\urlalt}[2]{\href{#1}{#2}}
\providecommand{\doi}[1]{doi:\urlalt{http://dx.doi.org/#1}{#1}}
\providecommand{\bibinfo}[2]{#2}

\bibitemdeclare{inproceedings}{Abbeel04}
\bibitem{Abbeel04}
\bibinfo{author}{Pieter Abbeel} \& \bibinfo{author}{Andrew~Y. Ng}
  (\bibinfo{year}{2004}): \emph{\bibinfo{title}{Apprenticeship learning via
  inverse reinforcement learning}}.
\newblock In: {\sl \bibinfo{booktitle}{International Conference on Machine
  Learning}}, pp. \bibinfo{pages}{1--}, \doi{10.1145/1015330.1015430}.

\bibitemdeclare{phdthesis}{Acar}
\bibitem{Acar}
\bibinfo{author}{Umut~A. Acar} (\bibinfo{year}{2005}):
  \emph{\bibinfo{title}{Self-adjusting computation}}.
\newblock Ph.D. thesis, \bibinfo{school}{Carnegie Mellon University},
  \bibinfo{address}{Pittsburgh, PA, USA}.

\bibitemdeclare{inproceedings}{Andre01}
\bibitem{Andre01}
\bibinfo{author}{David Andre} \& \bibinfo{author}{Stuart Russell}
  (\bibinfo{year}{2001}): \emph{\bibinfo{title}{Programmable Reinforcement
  Learning Agents}}.
\newblock In: {\sl \bibinfo{booktitle}{Advances in Neural Information
  Processing Systems}}, pp. \bibinfo{pages}{1019--1024}.

\bibitemdeclare{inproceedings}{Andre02}
\bibitem{Andre02}
\bibinfo{author}{David Andre} \& \bibinfo{author}{Stuart Russell}
  (\bibinfo{year}{2002}): \emph{\bibinfo{title}{State Abstraction for
  Programmable Reinforcement Learning Agents}}.
\newblock In: {\sl \bibinfo{booktitle}{Eighteenth National Conference on
  Artificial Intelligence}}, pp. \bibinfo{pages}{119--125}.

\bibitemdeclare{article}{Auer02}
\bibitem{Auer02}
\bibinfo{author}{Peter Auer}, \bibinfo{author}{Nicol\`{o} Cesa-Bianchi} \&
  \bibinfo{author}{Paul Fischer} (\bibinfo{year}{2002}):
  \emph{\bibinfo{title}{Finite-time Analysis of the Multiarmed Bandit
  Problem}}.
\newblock {\sl \bibinfo{journal}{Machine Learning}} \bibinfo{volume}{27}, pp.
  \bibinfo{pages}{235--256}, \doi{10.1023/A:1013689704352}.

\bibitemdeclare{article}{Bauer11}
\bibitem{Bauer11}
\bibinfo{author}{Tim Bauer}, \bibinfo{author}{Martin Erwig},
  \bibinfo{author}{Alan Fern} \& \bibinfo{author}{Jervis Pinto}
  (\bibinfo{year}{2011}): \emph{\bibinfo{title}{Adaptation-Based Programming in
  Java}}.
\newblock {\sl \bibinfo{journal}{PEPM '11}} , pp.
  \bibinfo{pages}{81--90}\doi{10.1145/1929501.1929518}.

\bibitemdeclare{misc}{haskell-abp}
\bibitem{haskell-abp}
\bibinfo{author}{{Bauer, Tim and Erwig, Martin and Fern, Alan and Pinto,
  Jervis}}: \emph{\bibinfo{title}{{ABP}}}.
\newblock \bibinfo{note}{\url{http://web.engr.oregonstate.edu/~bauertim/abp/}}.

\bibitemdeclare{book}{Bishop06}
\bibitem{Bishop06}
\bibinfo{author}{Christopher Bishop} (\bibinfo{year}{2006}):
  \emph{\bibinfo{title}{Pattern Recognition and Machine Learning}}.
\newblock \bibinfo{publisher}{Springer}.

\bibitemdeclare{inproceedings}{Dietterich98}
\bibitem{Dietterich98}
\bibinfo{author}{Thomas Dietterich} (\bibinfo{year}{1998}):
  \emph{\bibinfo{title}{The {MAXQ} Method for Hierarchical Reinforcement
  Learning}}.
\newblock In: {\sl \bibinfo{booktitle}{International Conference on Machine
  Learning}}, pp. \bibinfo{pages}{118--126}.

\bibitemdeclare{inproceedings}{Lagoudakis00}
\bibitem{Lagoudakis00}
\bibinfo{author}{Michail Lagoudakis} \& \bibinfo{author}{Michael Littman}
  (\bibinfo{year}{2000}): \emph{\bibinfo{title}{Algorithm Selection using
  Reinforcement Learning}}.
\newblock In: {\sl \bibinfo{booktitle}{International Conference on Machine
  Learning}}, pp. \bibinfo{pages}{511--518}.

\bibitemdeclare{article}{Lai85}
\bibitem{Lai85}
\bibinfo{author}{T.~Lai} \& \bibinfo{author}{H.~Robbins}
  (\bibinfo{year}{1985}): \emph{\bibinfo{title}{Asymptotically efficient
  adaptive allocation rules}}.
\newblock {\sl \bibinfo{journal}{Advances in Applied Mathematics}}
  \bibinfo{volume}{6}, pp. \bibinfo{pages}{4--22},
  \doi{10.1109/TAC.1987.1104491}.

\bibitemdeclare{article}{Lev44}
\bibitem{Lev44}
\bibinfo{author}{K.~Levenberg} (\bibinfo{year}{1944}): \emph{\bibinfo{title}{A
  method for the solution of certain non-linear problems in least squares}}.
\newblock {\sl \bibinfo{journal}{Applied Math Quarterly}} , pp.
  \bibinfo{pages}{164--168}.

\bibitemdeclare{inproceedings}{Littman94}
\bibitem{Littman94}
\bibinfo{author}{Michael Littman} (\bibinfo{year}{1994}):
  \emph{\bibinfo{title}{Markov Games as a Framework for Multi-Agent
  Reinforcement Learning}}.
\newblock In: {\sl \bibinfo{booktitle}{International Conference on Machine
  Learning}}, pp. \bibinfo{pages}{157--163}.

\bibitemdeclare{inproceedings}{Maclin05}
\bibitem{Maclin05}
\bibinfo{author}{R.~Maclin}, \bibinfo{author}{J.~Shavlik},
  \bibinfo{author}{L.~Torrey}, \bibinfo{author}{T.~Walker} \&
  \bibinfo{author}{E.~Wild} (\bibinfo{year}{2005}):
  \emph{\bibinfo{title}{Giving Advice about Preferred Actions to Reinforcement
  Learners Via Knowledge-Based Kernel Regression}}.
\newblock In: {\sl \bibinfo{booktitle}{Proceedings of the Twentieth National
  Conference on Artificial Intelligence}}, pp. \bibinfo{pages}{819--824}.

\bibitemdeclare{article}{MAR63}
\bibitem{MAR63}
\bibinfo{author}{D.~Marquardt} (\bibinfo{year}{1963}): \emph{\bibinfo{title}{An
  algorithm for least-squares estimation of nonlinear parameters}}.
\newblock {\sl \bibinfo{journal}{SIAM Journal of Applied Mathematics}} .

\bibitemdeclare{techreport}{UCTP}
\bibitem{UCTP}
\bibinfo{author}{H.~B. Nielsen} (\bibinfo{year}{2000}):
  \emph{\bibinfo{title}{UCTP - Test Problems for Unconstrained Optimization}}.
\newblock \bibinfo{type}{Technical Report}, \bibinfo{institution}{Technical
  University of Denmark}.

\bibitemdeclare{article}{Robbins52}
\bibitem{Robbins52}
\bibinfo{author}{H.~Robbins} (\bibinfo{year}{1952}): \emph{\bibinfo{title}{Some
  Aspects of the Sequential Design of Experiments}}.
\newblock {\sl \bibinfo{journal}{Bulletin of the American Mathematical
  Society}} \bibinfo{volume}{58}, pp. \bibinfo{pages}{527--535},
  \doi{10.1090/S0002-9904-1952-09620-8}.

\bibitemdeclare{inproceedings}{Ruv08}
\bibitem{Ruv08}
\bibinfo{author}{Paul Ruvolo}, \bibinfo{author}{Ian~R. Fasel} \&
  \bibinfo{author}{Javier~R. Movellan} (\bibinfo{year}{2008}):
  \emph{\bibinfo{title}{Optimization on a Budget: A Reinforcement Learning
  Approach}}.
\newblock In: {\sl \bibinfo{booktitle}{Neural Information Processing Symposium
  (NIPS)}}, pp. \bibinfo{pages}{1385--1392}.

\bibitemdeclare{inproceedings}{SPC08}
\bibitem{SPC08}
\bibinfo{author}{T.~Schrijvers}, \bibinfo{author}{S.~{Peyton-Jones}} \&
  \bibinfo{author}{M.~Chakravarty} (\bibinfo{year}{2008}):
  \emph{\bibinfo{title}{{Type Checking with Open Type Functions}}}.
\newblock In: {\sl \bibinfo{booktitle}{{ACM Int.\ Conf.\ on Functional
  Programming}}}, pp. \bibinfo{pages}{51--62}, \doi{10.1145/1411203.1411215}.

\bibitemdeclare{inproceedings}{Simpkins08}
\bibitem{Simpkins08}
\bibinfo{author}{Christopher Simpkins}, \bibinfo{author}{Sooraj Bhat},
  \bibinfo{author}{Michael Mateas} \& \bibinfo{author}{Charles Isbell}
  (\bibinfo{year}{2008}): \emph{\bibinfo{title}{Toward Adaptive Programming:
  Integrating Reinforcement Learning into a Programming Language}}.
\newblock In: {\sl \bibinfo{booktitle}{ACM Conference on Object-Oriented
  Programming Systems, Languages and Applications}}, pp.
  \bibinfo{pages}{603--614}, \doi{10.1145/1449955.1449811}.

\bibitemdeclare{book}{RL-book}
\bibitem{RL-book}
\bibinfo{author}{Richard Sutton} \& \bibinfo{author}{Andrew Barto}
  (\bibinfo{year}{2000}): \emph{\bibinfo{title}{Reinforcement Learning: An
  Introduction}}.
\newblock \bibinfo{publisher}{MIT Press}.

\bibitemdeclare{book}{Tho91}
\bibitem{Tho91}
\bibinfo{author}{S.~Thompson} (\bibinfo{year}{1991}):
  \emph{\bibinfo{title}{{Type Theory and Functional Programming}}}.
\newblock \bibinfo{publisher}{{Ad\-di\-son-Wes\-ley}},
  \bibinfo{address}{Redwood City, CA, USA}.

\end{thebibliography}
\end{document}